\documentclass[11pt,letterpaper]{article}
\usepackage[letterpaper,margin=1in]{geometry}

\usepackage{amsmath}
\usepackage{amssymb}
\usepackage{amsthm}

\newtheorem{theorem}[equation]{Theorem}
\newtheorem{lemma}[equation]{Lemma}
\newtheorem{definition}[equation]{Definition}
\newtheorem{hypothesis}[equation]{Hypothesis}

\newcommand{\N}{\mathbb{N}}
\newcommand{\eps}{\varepsilon}

\renewcommand{\le}{\leqslant}
\renewcommand{\ge}{\geqslant}

\renewcommand{\to}{\longrightarrow}

\DeclareMathOperator{\CG}{CG}
\DeclareMathOperator{\VG}{VG}
\DeclareMathOperator{\LCWIS}{LCWIS}
\DeclareMathOperator{\WLCWIS}{WLCWIS}

\newcommand{\sequence}[1]{\langle #1 \rangle}

\newcommand{\A}{\textnormal{A}}
\newcommand{\B}{\textnormal{B}}
\newcommand{\Y}{\textnormal{Y}}
\newcommand{\Z}{\textnormal{Z}}

\begin{document}

\title{
  Why is it hard to beat $O(n^2)$ for\\ Longest Common Weakly Increasing Subsequence?
}
\author{
  Adam Polak\thanks{
    This work was supported by the Polish Ministry of Science
    and Higher Education program \it{Diamentowy Grant}.
  }\\
  \small{Department of Theoretical Computer Science}\\
  \small{Faculty of Mathematics and Computer Science}\\
  \small{Jagiellonian University}\\
  \texttt{polak@tcs.uj.edu.pl}
}
\date{
}
\maketitle

\begin{abstract}
The Longest Common Weakly Increasing Subsequence problem (LCWIS) is a variant of
the classic Longest Common Subsequence problem (LCS).
Both problems can be solved with simple quadratic time algorithms.
A recent line of research led to a number of matching conditional lower bounds
for LCS and other related problems. However, the status of LCWIS remained open.

In this paper we show that LCWIS cannot be solved in $O(n^{2-\eps})$ time unless
the Strong Exponential Time Hypothesis (SETH) is false.

The ideas which we developed can also be used to obtain a lower bound based on
a safer assumption of $\mathsf{NC}$-SETH, i.e.~a version of SETH which talks
about $\mathsf{NC}$ circuits instead of less expressive CNF formulas.
\end{abstract}

\clearpage

\section{Introduction}

Despite attracting interest of many researches,
both from theoretical computer science and computational biology communities,
for many years the classic Longest Common Subsequence problem (LCS) have not
seen any significant improvement over the simple $O(n^2)$ dynamic programming
algorithm. The current fastest, $O(n^2/\log^2n)$ algorithm
by Masek and Paterson~\cite{MaPa80}, dates back to 1980.

Difficulties in making progress on the LCS inspired studying numerous related
problems, among them the Longest Common Increasing Subsequence problem (LCIS),
for which Yang, Huang, and Chao~\cite{Yang05}
found a quadratic time dynamic programming algorithm.
Their algorithm was later improved by Sakai~\cite{Sakai} to work in linear space.
Even though both these algorithms are devised to compute the Longest Common
Increasing Subsequence, they can be easily modified to compute the
Longest Common Weakly Increasing Subsequence (LCWIS).
The latter problem, first introduced by Kutz et al.~\cite{Kutz11},
can be solved in linear time in the special case of a $3$-symbols alphabet,
as proposed by Duraj~\cite{Duraj13}.
However, despite some attempts over the last decade, no subquadratic time
algorithm has been found for the general case of LCWIS.

A recent line of research led to a number of conditional lower bounds for
polynomial time solvable problems. In particular
Abboud, Backurs, and Vassilevska Williams~\cite{Abboud15},
and independently Bringmann and K\"unnemann~\cite{Bringmann15}
proved that LCS cannot be solved in $O(n^{2-\eps})$ time unless the Strong
Exponential Time Hypothesis (SETH) is false.

\begin{hypothesis}[Strong Exponential Time Hypothesis]
There is no $\eps > 0$ such that for all $k \ge 3$, $k$-SAT on $N$ variables
can be solved in $O(2^{(1-\eps)N})$ time.
\end{hypothesis}

Moreover, Bringmann and K\"unnemann~\cite{Bringmann15} proposed a general framework
for proving quadratic time hardness of sequence similarity measures. Within this
framework, it is sufficient to show that a similarity measure
\emph{admits an alignment gadget} to prove that this similarity measure cannot be
computed in $O(n^{2-\eps})$ time unless SETH is false.
Besides LCS many other similarity measures, e.g.~Edit Distance and Dynamic Time
Warping, fall into this framework.
However, it seems that neither LCIS nor LCWIS admits an alignment gadget.

In this paper we show that LCWIS cannot be solved in $O(n^{2-\eps})$ time unless
SETH is false. We do this by proving the following theorem.

\begin{theorem}
\label{th:max-cnf-sat}
If the Longest Common Weakly Increasing Subsequence problem for two sequences
of length $n$ can be solved in $O(n^{2-\eps})$ time,
then given a CNF formula on $N$ variables and $M$ clauses
it is possible to compute the maximum number of satisfiable clauses
(MAX-CNF-SAT) in $O(2^{(1-\eps/2)N} \mbox{poly}(M))$ time.
\end{theorem}

Our reduction is modelled after previous hardness results based on SETH,
in particular~\cite{Abboud15} and~\cite{Backurs15}.
We go through the Most-Orthogonal Vectors problem, and construct
\emph{vector gadgets} such that two vector gadgets have large LCWIS iff
the corresponding vectors have small inner product.
The crucial ingredient is a construction that lets us combine many vector
gadgets into two sequences such that their LCWIS depends on the largest LCWIS
among all pairs of vector gadgets.

Unlike $\mathsf{P}\neq\mathsf{NP}$ and several other common assumptions
for conditional lower bounds in computational complexity, SETH is considered
by many not a very safe working hypothesis.
Recently, Abboud et al.~\cite{Abboud16} came up with a weaker assumption,
which still allows to prove many previous SETH-based lower bounds.
More specifically, they propose a reduction from satisfiability of
\emph{nondeterministic Branching Programs}~\cite{Arora} (BP-SAT) to LCS and any
other similarity measure which admits an alignment gadget.
Their reduction implies that the existence of a strongly subquadratic time
algorithm for LCS not only refutes SETH but also has stronger consequences,
e.g.~the existence of an~$O(2^{(1-\delta)N})$ time algorithm for satisfiability
of $\mathsf{NC}$ circuits. For an in-depth discussion of consequences of their
reduction and motivations to study such reductions please refer to the original
paper~\cite{Abboud16}. At the end of our paper we briefly explain how the ideas
which we developed can be used to obtain a reduction from BP-SAT to LCWIS.

Unfortunately, our techniques are not sufficient to prove a similar lower bounds
for LCIS.

\section{Preliminaries}

Let us start with the formal definition of the LCWIS problem.

\begin{definition}[Longest Common Weakly Increasing Subsequence]
Given two sequences $A$ and $B$ over an alphabet $\Sigma$ with a linear
order $\le_\Sigma$, the \emph{Longest Common Weakly Increasing Subsequence}
problem asks to find a sequence $C$ such that
\begin{itemize}
  \item it is weakly increasing with respect to~$\le_\Sigma$,
  \item it is a subsequence of both $A$ and $B$,
  \item and its length is maximum possible.
\end{itemize}
We denote the length of $C$ by $\LCWIS(A, B)$.
\end{definition}

\noindent
For example, $\LCWIS(\sequence{1,2,5,2,5,3}, \sequence{2,4,5,2,3,4}) = 3$,
and the optimal subsequence is $\sequence{2,2,3}$.

To simplify further arguments we introduce, as an auxiliary problem, the
weighted version of LCWIS.

\begin{definition}[Weighted Longest Common Weakly Increasing Subsequence]
Given two sequences $A$ and $B$ over an alphabet $\Sigma$ with a linear
order $\le_\Sigma$ and the \emph{weight function} $w:\Sigma \to \N_+$,
the \emph{Weighted Longest Common Weakly Increasing Subsequence} (WLCWIS)
problem asks to find a sequence $C$ such that
\begin{itemize}
  \item it is weakly increasing with respect to $\le_\Sigma$,
  \item it is a subsequence of both $A$ and $B$,
  \item and its total weight, i.e.~$\sum_{i=1}^{|C|}w(C_i)$, is maximum possible.
\end{itemize}
We denote the total weight of $C$ by $\WLCWIS(A, B)$.
\end{definition}

\begin{lemma}
\label{th:wlcwis-to-lcwis}
For a sequence $X = \sequence{X_1, X_2, \ldots, X_{|X|}}$
let $\widehat{X}$ denote a sequence obtained from $X$ by replacing each symbol
$a$ by its $w(a)$ many copies, i.e.
$$\widehat{X} = X_1^{w(X_1)}\ X_2^{w(X_2)}\ \ldots\ X_{|X|}^{w(X_{|X|})}.$$

Computing the $\WLCWIS$ of two sequences $A$ and $B$ of total weight
at most $n$ can be reduced to computing the $\LCWIS$ of two sequences
$\widehat{A}$ and $\widehat{B}$ of length at most $n$.
\end{lemma}

\newcommand{\wA}{\widehat{A}}
\newcommand{\wB}{\widehat{B}}

\begin{proof}
We have to show that $\WLCWIS(A, B) = \LCWIS(\wA, \wB)$.
Every common weakly increasing subsequence $C$ of $A$ and $B$ translates
to a common weakly increasing subsequence $\widehat{C}$ of $\wA$ and $\wB$,
and the length of $\widehat{C}$ equals the total weight of $C$, thus
$\WLCWIS(A, B) \le \LCWIS(\wA, \wB)$.

The proof of inequality $\WLCWIS(A, B) \ge \LCWIS(\wA, \wB)$ is by induction
on the alphabet size. Denote by $C$ the longest common weakly increasing
subsequence of $\wA$ and $\wB$, by $a$ the smallest symbol in the alphabet,
and by $k$ the number of occurrences of symbol $a$ at the beginning of $C$.
Note that $k$ must be a multiple of $w(a)$, because otherwise $C$ could be
extended. To finish the proof, cut off the shortest prefixes of $A$ and $B$
containing at least $k/w(a)$ occurrences of symbol $a$, remove all symbols $a$
from remaining suffixes, and apply inductive hypothesis to get a sequence
that can be appended to $a^{k/w(a)}$ to get a common weakly increasing
subsequence of $A$ and $B$ of the desired total weight.
\end{proof}

The curious reader looking for a more detailed argument is referred
to~\cite{Abboud15}, where an analogous lemma for the weighted version of LCS is
presented. Its proof can be directly translated to another proof for the case of
WLCWIS.

\section{Reduction from MAX-CNF-SAT to LCWIS}

This section is devoted to proving Theorem~\ref{th:max-cnf-sat}.
We do this by showing a reduction from the Most-Orthogonal Vectors problem,
introduced by Abboud, Backurs, and Vassilevska Williams~\cite{Abboud15}.

\begin{definition}[Most-Orthogonal Vectors]
Given two sets of vectors $U, V \subseteq \{0,1\}^d$, both of the same size $n$,
and an integer $r \in \{0, 1, \ldots, d\}$, are there two vectors
$u \in U$ and $v \in V$ such that their inner product does not exceed $r$,
i.e.~$u \cdot v := \sum_{i=1}^{d} u_i \cdot v_i \le r$?
\end{definition}

\begin{lemma}[Abboud, Backurs, Vassilevska Williams~\cite{Abboud15}]
\label{th:sat-to-ovp}
If Most-Orthogonal Vectors on $n$ vectors in $\{0,1\}^d$ can be solved in
$T(n, d)$ time, then given a $CNF$ formula on $N$ variables and $M$ clauses,
we can compute the maximum number of satisfiable clauses (MAX-CNF-SAT),
in $O(T(2^{N/2}, M) \cdot \log M)$ time.
\end{lemma}

In our reduction we use alphabet $\Sigma$ consisting of the integers from
$3$ to $3d+2$.
First we define two kinds of \emph{coordinate gadgets}:
$\CG_1, \CG_2 : \{0,1\} \times \{1,2,\ldots,d\} \to \Sigma^*$.
\begin{align*}
\CG_1(0, i) &= \sequence{3i, 3i+1}\\
\CG_1(1, i) &= \sequence{3i+2}\\
\CG_2(0, i) &= \sequence{3i, 3i+2}\\
\CG_2(1, i) &= \sequence{3i+1}
\end{align*}
Observe that
$$LCWIS(\CG_1(x, i), \CG_2(y, i)) = \begin{cases}
  0,& \text{if } x = 1 \text{ and } y = 1,\\
  1,& \text{otherwise.}
\end{cases}$$

Now we can define two kinds of \emph{vector gadgets}:
$\VG_1, \VG_2 : \{0,1\}^d \to \Sigma^*$.
\begin{align*}
\VG_1(u_1, u_2, \ldots, u_d) &=
  \CG_1(u_1, 1)\ \CG_1(u_2, 2)\ \ldots\ \CG_1(u_d, d) \\
\VG_2(u_1, u_2, \ldots, u_d) &=
  \CG_2(u_1, 1)\ \CG_2(u_2, 2)\ \ldots\ \CG_2(u_d, d)
\end{align*}
Observe that
$$LCWIS(\VG_1(u), \VG_2(v)) = d - (u \cdot v).$$

The next lemma helps to us combine many vector gadgets into two sequences
such that computing their WLCWIS lets us find two vector gadgets with the
largest LCWIS, which corresponds to the pair of most orthogonal vectors.

\begin{lemma}
\label{th:main-lemma}
Let $S=\{s_1, s_2, \ldots, s_n\} \subseteq \Sigma^*$, $T=\{t_1, t_2, \ldots, t_n\} \subseteq \Sigma^*$.
Augment the alphabet $\Sigma$ with four additional symbols $\A, \B, \Y, \Z$,
such that $\A < \B < \sigma < \Y < \Z$ for all $\sigma \in \Sigma$.
Denote by $\ell$ the maximum total weight of any sequence in $S \cup T$,
and let the weights of the new symbols be $w(\A) = w(\Z) = \ell$
and $w(\B) = w(\Y) = 2\ell$.
Finally, define two sequences $P_1$ and $P_2$,
$$P_1 = \A^{2n}\ s_1\ \Y\B\ s_2\ \Y\B\ \ldots\ \Y\B\ s_n\ \Z^{2n},$$
$$P_2 = (\Z\Y\B\A)^{n}\ t_1\ \Z\Y\B\A\ t_2\ \Z\Y\B\A\ \ldots\ \Z\Y\B\A\ t_n\ (\Z\Y\B\A)^{n}.$$
Then
$$
\WLCWIS(P_1, P_2) =
\max_{\substack{1 \le i, j \le n}} \WLCWIS(s_i, t_j) + (4n-2)\cdot\ell.
$$
\end{lemma}

\begin{proof}
The proof consists of two parts, first we prove that a common weakly increasing
subsequence of $P_1$ and $P_2$ with the claimed total weight exists,
then we prove that there is no such subsequence with a larger total weight.

Take $i$ and $j$ maximizing $\WLCWIS(s_i, t_j)$,
denote by $Q$ the corresponding subsequence of $s_i$ and $t_j$,
and consider the following sequence:
$$Q' = \A^{n+j-1-(i-1)}\ \B^{i-1}\ Q\ \Y^{n-i}\ \Z^{2n-j-(n-i)}.$$
The construction of $Q'$ is motivated by the following intuition.
We start with $Q$, and append at the beginning one $\B$ for each $\Y\B$-fragment
appearing before $s_i$ in $P_1$. There are $i-1$ such fragments.
On the other hand, the number of $\Z\Y\B\A$-fragments appearing before $t_j$
in $P_2$ equals $n+j-1$.
We devote the rightmost $i-1$ of these fragments to find matching $\B$
symbols, and for each of the remaining fragments we take one
$\A$ and append it at the beginning of $Q'$. The $\A^{2n}$ prefix of $P_1$ is
long enough to match all these $\A$ symbols.
In an analogous way we append $\Y$ and $\Z$ elements at the end of $Q'$.

Note that $Q'$ is weakly increasing and it appears in both $P_1$ and $P_2$ as
a subsequence. The total weight of $Q'$ equals
\begin{multline*}
\ell (n+j-1-(i-1)) + 2\ell (i-1) + \WLCWIS(s_i, t_j) +
2\ell (n-i) + \ell (2n-j-(n-i)) = \\
= \ell \cdot 2n + 2\ell \cdot (n-1) + \WLCWIS(s_i, t_j) = \WLCWIS(s_i, t_j) + (4n-2) \cdot \ell.
\end{multline*}
This finishes the first part of the proof.

Now we prove that there is no common weakly increasing subsequence of
$P_1$ and $P_2$ with a larger total weight. Consider any such subsequence.
Denote by $c_S$ the number of $s_i$-fragments in $P_1$ that contribute at least
one symbol to that subsequence. Analogously, denote by $c_T$ the number of
$t_j$-fragments in $P_2$ contributing at least one symbol.
Finally, denote by $c_\A, c_\B, c_\Y, c_\Z$ the number of
symbols $\A, \B, \Y, \Z$ in the subsequence, respectively.
Observe that the contribution of $\A, \B, \Y, \Z$ symbols to the total weight
of the subsequence is
\begin{multline}
\label{eq:abyz}
c_\A \cdot w(\A) + c_\B \cdot w(\B) + c_\Y \cdot w(\Y) + c_\Z \cdot w(\Z) = \\
= c_\A \cdot \ell + c_\B \cdot 2\ell + c_\Y \cdot 2\ell + c_\Z \cdot \ell = \\
= (c_\A + c_\B + c_\Y + c_\Z) \cdot \ell + (c_\B + c_\Y) \cdot \ell.
\end{multline}
Now the proof splits into two cases:

\paragraph*{Case 1: $c_S \le 1$ and $c_T \le 1$.}
The contribution of $s_i$- and $t_j$-fragments to the total weight of the
subsequence is at most $\max_{\substack{1 \le i, j \le n}} \WLCWIS(s_i, t_j)$.
The remaining weight of the subsequence must come from the symbols
$\A, \B, \Y, \Z$.
Each $\Z\Y\B\A$-fragment in $P_2$ can contribute at most one symbol to the,
weakly increasing, subsequence. There are $3n - 1$ such fragments, therefore
$$c_\A + c_\B + c_\Y + c_\Z \le 3n - 1.$$
Similarly, each $\Y\B$-fragment in $P_1$ can contribute at most one symbol.
There are $n - 1$ such fragments in $P_2$, therefore
$$c_\B + c_\Y \le n - 1.$$
Combining these two inequalities with the equation~(\ref{eq:abyz}) we get that
the contribution of $\A, \B, \Y, \Z$ symbols to the total weight of the
subsequence is at most
$(3n - 1) \cdot \ell + (n - 1) \cdot \ell = (4n - 2) \cdot \ell$.
This together with the contribution of the $s_i$- and $t_j$-fragments
gives the required upper bound for $\WLCWIS(P_1, P_2)$.
\paragraph*{Case 2: $c_S \ge 2$ or $c_T \ge 2$.}
The contribution of $s_i$- and $t_j$-fragments to the total weight of the
subsequence is at most $\min(c_S, c_T) \cdot \ell$.
Since $\min(c_S, c_T) = c_S + c_T - \max(c_S, c_T)$ and at least one of
$c_S, c_T$ is at least $2$, we can bound $\min(c_S, c_T)$ from above by
$c_S + c_T - 2$.
Thus the considered contribution is at most
$$((c_S - 1) + (c_T - 1))\cdot\ell.$$
If two $s_i$-fragments contribute at least one symbol each,
no $\Y\B$-fragment between them can contribute,
because the subsequence has to be weakly increasing.
Therefore, the number of $\Y\B$-fragments contributing is at most
$n - 1 - (c_S - 1)$, thus
$$c_\B + c_\Y \le n - 1 - (c_S - 1).$$
Similarly
$$c_\A + c_\B + c_\Y + c_\Z \le 3n - 1 - (c_T - 1).$$
Using the equation~(\ref{eq:abyz}) we get that the total weight of
the subsequence is at most
$$((c_S - 1) + (c_T - 1)) \cdot \ell +
(3n - 1 - (c_T - 1)) \cdot \ell +
(n - 1 - (c_S - 1)) \cdot \ell,$$
which is $(4n - 2) \cdot \ell$.
\end{proof}

Now we are ready to show our reduction from Most-Orthogonal Vectors to LCWIS.

\begin{lemma}
\label{th:ovp-to-lcwis}
If LCWIS of two sequences of length $n$ can be computed in $O(n^{2-\eps})$ time,
then Most-Orthogonal Vectors for $n$ vectors in $\{0,1\}^d$
problem can be solved in $O((dn)^{2-\eps})$.
\end{lemma}

\begin{proof}
Given two sets of vectors $U$ and $V$, we use $VG_1$ to construct vector gadgets
for vectors from set $U$, and $VG_2$ to construct vector gadgets for vectors
from set $V$. We get two sets, each containing $n$ vector gadgets of length at
most $2d$. Now we put weight $1$ for every symbol used to construct the vector
gadgets, and use Lemma~\ref{th:main-lemma} to construct two sequences, of total
weight $O(dn)$, such that their WLCWIS can be used to calculate the inner
product of the pair of most orthogonal vectors. Finally, we use
Lemma~\ref{th:wlcwis-to-lcwis} to reduce computing the WLCWIS of these two
sequences to computing LCWIS of two sequences of length $O(dn)$.
\end{proof}

\begin{proof}[Proof of Theorem~\ref{th:max-cnf-sat}]
This theorem is a direct conclusion from Lemma~\ref{th:sat-to-ovp}
and Lemma~\ref{th:ovp-to-lcwis}.
\end{proof}

\section{Reduction from BP-SAT to LCWIS}
Abboud et al.~\cite{Abboud16} prove that if LCS for two sequences of length $n$
can be computed in $O(n^{2-\eps})$ time, then given a nondeterministic Branching
Program of length $T$ and width $W$ on $N$ variables it is possible to decide if
there is an assignment to the variables that makes the program accept in
$O(2^{(1-\eps/2)N} \mbox{poly}(T^{\log W}))$ time.
The existence of such an algorithm would have much more remarkable consequences
in computational complexity than just refuting SETH.
For an in-depth discussion of these consequences we refer the curious reader
to~\cite{Abboud16}.

The above-mentioned reduction exploits properties of LCS in three ways:
when reducing LCS-Pair to LCS (Lemma~1 in~\cite{Abboud16}),
when constructing the AND gadgets (Lemma~2 in~\cite{Abboud16}),
and when constructing the OR gadgets (Lemma~3 in~\cite{Abboud16}).

Our techniques can be used to show a similar reduction to LCWIS.
We follow the general construction of~\cite{Abboud16}, and replace parts
specific to LCS with LCWIS analogues. Lemma~1 in~\cite{Abboud16} is replaced
with our Lemma~\ref{th:main-lemma}. The AND gadget for LCWIS can be obtained by
concatenating the corresponding sequences and using the disjoint union of their
alphabets as the alphabet for concatenation. In the new alphabet all symbols
from the first sequence should be made smaller than all symbols from the second
sequence. To obtain the OR gadget we use the same disjoint union construction
for the alphabet, and apply our Lemma~\ref{th:main-lemma} to resulting
sequences. Finally, we follow the proof of Theorem~3 in~\cite{Abboud16} to glue
these three pieces together and get the reduction from BP-SAT to LCWIS.

\section{Open problems}

The reductions which we use in this paper do not work with a constant size alphabet.
In the case of LCS it is possible to obtain conditional quadratic time lower
bounds even for binary strings~\cite{Bringmann15}. It is very unlikely
to obtain similar results for LCWIS given the linear time algorithm for the
$3$-letter alphabet~\cite{Duraj13}. However, it is an open problem to show
quadratic time hardness of LCWIS for an alphabet of some larger constant size.

Even though $O(n^2)$ algorithms for LCIS and LCWIS are virtually identical,
our reductions cannot be easily modified to work with LCIS.
The status of LCIS is thus open. It is possible that it can be solved in
strongly subquadratic time, but it is also possible that there is a reduction
proving that such an algorithm is unlikely.

The existing techniques for constructing reductions appear insufficient for both
of these open problems and some new ideas seem necessary.

\bibliographystyle{acm}
\bibliography{lcwis}

\end{document}